\documentclass[aip,%onecolumn,
    secnumarabic,
    amssymb, amsmath, longbibliography,
    jmp, floatfix, preprint,tightenlines,a4paper]{revtex4-1}
\pdfoutput=1
\usepackage[utf8]{inputenc}
\usepackage[english]{babel}
\usepackage{amsmath}
\usepackage{amsfonts}
\usepackage{amsthm}
\usepackage{graphicx}
\usepackage{physics}
\usepackage{mathtools}
\usepackage{xspace}
\usepackage{tikz-cd}
\usepackage{breakurl}
\usepackage{hyperref}

\usepackage[toc,page]{appendix}

\usepackage{enumitem}

\setcitestyle{numbers,open={[},close={]}}

\DeclarePairedDelimiter{\ceil}{\lceil}{\rceil}
\DeclarePairedDelimiter{\floor}{\lfloor}{\rfloor}
\DeclarePairedDelimiter{\inn}{\langle}{\rangle}

\theoremstyle{definition}
\newcounter{counter}
\newtheorem{theorem}{Theorem}
\newtheorem{definition}{Definition}
\newtheorem{proposition}[counter]{Proposition}
\newtheorem{lemma}[counter]{Lemma}
\newtheorem{corollary}[counter]{Corollary}

\newtheorem{remark}[counter]{Remark}
\newtheorem*{theorem*}{Theorem}
\newtheorem*{lemma*}{Lemma}
\newtheorem*{definition*}{Definition}

\newcommand{\rnk}{\text{rnk}\xspace}
\newcommand{\R}{\mathbb{R}}

\newcommand{\N}{\mathbb{N}}

\newcommand{\id}{\text{id}}

\newcommand{\sa}{\text{sa}}
\usepackage{chngcntr}
\counterwithout{counter}{section}

\newcommand{\mult}{\,\&\,}
\newcommand{\commu}{\,\lvert\,}

\begin{document}
\author{John van de Wetering}
\email{john@vdwetering.name}
\affiliation{Institute for Computing and Information Sciences, Radboud University, Toernooiveld 212 Nijmegen, Netherlands}

\title{Sequential Product Spaces are Jordan Algebras}
\date{February 22, 2019}
\revised{April 24, 2019}

\begin{abstract}
\noindent We show that finite-dimensional order unit spaces equipped with a continuous sequential product as defined by Gudder and Greechie are homogeneous and self-dual. As a consequence of the Koecher-Vinberg theorem these spaces therefore correspond to Euclidean Jordan algebras. We remark on the significance of this result in the context of reconstructions of quantum theory. 
In particular, we show that sequential product spaces must be C$^*$-algebras when their vector space tensor product is also a sequential product space (in the parlance of operational theories, when the space `allows a local composite'). 
We also show that sequential product spaces in infinite dimension correspond to JB-algebras when a few additional conditions are satisfied. Finally, we remark on how changing the axioms of the sequential product might lead to a new characterisation of homogeneous cones.
\end{abstract}

\maketitle

\section{Introduction}

The set of observables of a quantum system can be represented as the space of self-adjoint operators on a complex Hilbert space $B(H)^{\text{sa}}$. This space has a variety of algebra-like structures that can be associated to it, the most well-known of which is the \emph{Jordan product} $a*b := \frac12(ab+ba)$. In the 30's Jordan, von Neumann and Wigner hoped to find generalisations of the quantum mechanical formalism by considering general spaces equipped with an axiomatisation of this algebraic structure. They however found that the resulting \emph{Euclidean Jordan algebras} (EJAs) have a strikingly simple classification \cite{jordan1993algebraic}, and hence that this algebraic approach does not allow you to go far beyond quantum theory. The significance of EJAs was further established by the Koecher-Vinberg theorem that states that any finite-dimensional homogeneous and self-dual ordered vector space is a Euclidean Jordan algebra \cite{koecher1957positivitatsbereiche}. These two results, the Koecher-Vinberg theorem and the classification by Jordan, von Neumann and Wigner, lie at the heart of many \emph{reconstructions of quantum theory} where intuitively sensible axioms from which quantum theory can be derived are sought \cite{hardy2011reformulating,wetering2018reconstruction,selby2018reconstructing,barnum2014higher,gunson1967algebraic} (although it should be noted that these theorems are not used directly in all approaches \cite{tull2016reconstruction,hohn2017toolbox,chiribella2011informational}).

Other algebraic structures on $B(H)^{sa}$ studied in an axiomatic way are the \emph{quadratic~Jordan~algebras} that axiomatize the map $(a,b)\mapsto aba$ or the more general \emph{triple~product} $(a,b,c)\mapsto \frac{1}{2}(abc+cba)$. The definitions of the Jordan product and the triple product don't have a particularly compelling physical motivation: the product does not correspond to any type of physical process. In this paper we will look at a different structure that \emph{does} follow naturally from physical processes. 
Given two positive operators $a$ and $b$ we define their \emph{sequential product} as $a\mult b := \sqrt{a}b\sqrt{a}$. When $a$ and $b$ represent \emph{effects}, i.e.\ possible outcomes in a measurement, then the sequential product models the act of first getting the outcome $a$ and then the outcome $b$, hence the name \emph{sequential} product (although this composition is also known as a \emph{L\"uders process}). It is important to note that this product is only defined for positive operators (since otherwise the square root wouldn't be defined), and that this operation is not bilinear, associative or commutative. Gudder and Greechie introduced the concept of a \emph{sequential~effect~algebra} to study the sequential product in a more abstract setting \cite{gudder2002sequential}. While they studied the structure of the sequential product on the very general structure of \emph{effect~algebras}, we will restrict ourselves to the more concrete setting of \emph{order~unit~spaces}:

\begin{definition}
    An \emph{order unit space} $(V, \leq, 1)$ is an ordered real vector space with the additional property that $1$ is a \emph{strong Archimedean unit}:
    \begin{enumerate}
        \item Strong unit: For all $a\in V$ we can find $n\in \N$ such that $-n1\leq a \leq n$.
        \item Archimedean: For $a\in V$ when $a\leq \frac1n 1$ for all $n\in\N_{>0}$, then $a\leq 0$.
    \end{enumerate}
    We call the elements $a\in V$ with $0\leq a \leq 1$ the \emph{effects} of $V$ which we will denote by $[0,1]_V$. The \emph{states} of $V$ are positive linear maps $\omega:V\rightarrow \R$ such that $\omega(1) = 1$.
\end{definition}

Ordered vector spaces with a strong unit represent the most general kinds of systems allowed in causal \emph{generalised probabilistic theories} \cite{barrett2007information} and hence they form a suitable background to studying models of general physical theories. The Archimedity condition states intuitively that there are no effects that cannot be distinguished by a state. More precisely, order unit spaces are precisely the ordered vector spaces where the states order-separate the elements: if $\omega(v)\leq \omega(w)$ for all states $\omega$ then $v\leq w$.

Note that an order unit space has a norm induced by the order in the following way: ${\norm{a} := \inf\{r\in \R_{\geq 0}~;~ -r1\leq a\leq r1\}}$. Whenever we refer to continuity in the context of order unit spaces it should be understood to refer to this norm.

The object of study in this paper is an order unit space with an operation modelled after the sequential product on $B(H)^{\text{sa}}$. To be specific:

\begin{definition}\label{def:seqprod}
    Let $(V, \leq, 1, \&)$ be an order unit space equipped with a binary operation \\ ${\&:[0,1]_V\times [0,1]_V \rightarrow [0,1]_V}$. We write $a\commu b$ and say $a$ and $b$ are \emph{compatible} when $a\mult b = b\mult a$. We call $V$ a \emph{sequential product space} and $\&$ a \emph{sequential product} when $\&$ satisfies the following properties for all $a,b,c \in [0,1]_V$:
    \begin{enumerate}[label=({S}\theenumi), ref=S\theenumi]
        \item \label{ax:add} Additivity: $a\mult (b+c) = a\mult b+ a\mult c$.
        \item \label{ax:cont} Continuity: The map $a\mapsto a\mult b$ is continuous in the norm.
        \item \label{ax:unit} Unitality: $1\mult a = a$.
        \item \label{ax:orth} Compatibility of orthogonal effects: If $a\mult b = 0$ then also $b\mult a =0$.
        \item \label{ax:assoc} Associativity of compatible effects: If $a\commu b$ then $a\mult (b\mult c) = (a\mult b)\mult c$.
        \item \label{ax:compadd} Additivity of compatible effects: If $a\commu b$ then $a \commu 1-b$, and if also $a\commu c$ then $a\commu (b+c)$.
        \item \label{ax:compmult} Multiplicativity of compatible effects: If $a\commu b$ and $a\commu c$ then $a\commu (b\mult c)$.
    \end{enumerate}
\end{definition}
The properties we require of $\&$ are the same as that of a sequential product in a sequential effect algebra \cite{gudder2002sequential} except for condition \ref{ax:cont} which is new. It should be noted that the standard sequential product $a\mult b = \sqrt{a}b\sqrt{a}$ on $B(H)^{\text{sa}}$ is not fully characterised by these axioms, as there are multiple binary operations that satisfy these axioms \cite{weihua2009uniqueness}. It is possible however to characterise the standard sequential product using related sets of axioms \cite{gudder2008characterization,westerbaan2016universal,wetering2018characterisation}.
The main examples of sequential product spaces are Euclidean Jordan algebras:
\begin{definition}
    We call a real vector space $V$ a \emph{Jordan algebra} when it has a bilinear commutative operation $*$ that satisfies the \emph{Jordan identity}: $a*(b*(a*a)) = (a*b)*(a*a)$. We call $V$ a \emph{Euclidean} Jordan algebra (EJA) when it is furthermore a finite-dimensional Hilbert space with inner product~$\inn{\cdot, \cdot}$ such that $\inn{a*b,c} = \inn{b,a*c}$.
\end{definition}
Euclidean Jordan algebras were originally introduced as generalisations of the quantum mechanical formalism under the name \emph{formally real} Jordan algebras~\cite{jordan1933}. Due to their connection to self-dual homogeneous cones, EJAs have been studied in many different contexts as well. The main examples of Euclidean Jordan algebras are the sets of self-adjoint operators $B(H)^{\text{sa}}$ where $H$ is a real, complex or quaternionic Hilbert space where the Jordan product is given by $A*B:=\frac12(AB+BA)$. In fact, as shown by the classification theorem of EJAs~\cite{jordan1993algebraic}, any EJA can be constructed as a direct sum of these algebras, a family of algebras called \emph{spin-factors}, and the \emph{exceptional} algebra.

It is not obvious how EJAs form sequential product spaces. Using the \emph{quadratic} map, which is in turn defined in terms of the Jordan product, a sequential product can be constructed~\cite{wetering2018characterisation}. The main purpose of this paper is to establish that EJAs are in fact the most general type of (finite-dimensional) sequential product space:
\begin{theorem*}
    Let $V$ be a finite-dimensional sequential product space. Then $V$ is order-isomorphic to a Euclidean Jordan algebra.
\end{theorem*}
\noindent This result is given as Theorem~\ref{theor:seqprodkoechervinberg} in Section~\ref{sec:subselfdual}. Since EJAs are very well understood and in particular classified we can use this theorem to prove additional results. In particular, using a property called \emph{local tomography} we can infer when sequential product spaces are C$^*$-algebras.
\begin{definition}\label{def:localcomposite}
    Let $V$ and $W$ be finite-dimensional sequential product spaces. We say that they have a \emph{locally tomographic composite} when their vector space tensor product $V\otimes W$ is also a sequential product space with $(a_1\otimes b_1)\mult (a_2\otimes b_2) = (a_1\mult a_2)\otimes (b_1\mult b_2)$ for all effects $a_i$ in $V$ and $b_i$ in $W$.
\end{definition}
\noindent In the context of generalised probabilistic theories, the property of local tomography states that local measurements on each of the subsystems is enough to fully characterise bipartite states. It is a property that holds for regular quantum theory, but fails for, for instance, quantum theory over the real numbers.
\begin{theorem*}
    Let $V$ be a finite-dimensional sequential product space that has a locally tomographic composite with itself, then there exists a C$^*$-algebra $A$ such that $V$ is isomorphic to  $A^{\text{sa}}$ as a Jordan algebra.
\end{theorem*}
\noindent This result is given as Theorem~\ref{theor:seqprodlocalcomp} in Section~\ref{sec:loctom}. Note that C$^*$-algebras being singled out among all the EJAs by local tomography is not surprising as similar results were obtained in Refs.~\cite{barnum2014local,selby2018reconstructing,masanes2014entanglement,hanche1985jb}, but in combination with the results regarding the sequential product it does give a novel understanding of the mathematical structure of quantum theory:

\begin{quote}
    A causal probabilistic physical theory that satisfies local tomography and that has a well-behaved notion of sequential measurement must be modelled by C$^*$-algebras.
\end{quote}

\noindent There have been quite a few characterisations of quantum theory using operational axioms~\cite{barnum2014higher,selby2018reconstructing,masanes2014entanglement,chiribella2011informational,tull2016reconstruction,wetering2018reconstruction,hohn2017toolbox,hardy2001quantum}, but the one presented above is different in a couple of ways. First of all, other characterisations and reconstructions have their axioms refer to a multitude of structures, like the existence of certain systems, transformations and pure states, instead of focusing on a single aspect (in this case, sequential measurement). 
Second, all reconstructions of quantum theory that the author is aware of have axioms ensuring the existence of suitable reversible (i.e.\ invertible) dynamics in the theory. In contrast, this characterisation of C$^*$-algebras doesn't directly say anything about the existence of reversible maps.

The proof we use for the correspondence between sequential product spaces and EJAs is highly specific to finite-dimensional spaces. Yet, by assuming a few other properties, we can also find a correspondence between $\sigma$-directed complete sequential product spaces and JB-algebras, which form an infinite-dimensional generalisation of Euclidean Jordan algebras. We refer to section \ref{sec:infiniterank} for the details.

The structure of the paper is as follows: the main theorem that finite-dimensional sequential product spaces are EJAs will be proved using the Koecher-Vinberg theorem which requires us to show that the space is both \emph{homogeneous} and \emph{self-dual}. In Section~\ref{sec:prelim} we will cover known results originally presented in~\cite{gudder2002sequential,wetering2018characterisation} regarding sequential product spaces, culminating in a proof of a spectral theorem and a proof of the homogeneity of the space. Then in section \ref{sec:selfdual} we will prove self-duality of the space, using results regarding lattices of projections of Alfsen and Shultz~\cite{alfsen2012state,alfsen2012geometry} and a characterisation result concerning low rank homogeneous spaces of Ito and Louren{\c c}o \cite{ito2017p}. At this point sequential product spaces have been established to be EJAs, but only in a rather indirect way. In section \ref{sec:jordanproduct} we directly construct the Jordan product using the sequential product. In Section~\ref{sec:loctom} we show how the additional requirement of local tomography forces the sequential product space to be a C$^*$-algebra. Section~\ref{sec:infiniterank} establishes an infinite-dimensional generalisation of the main theorem, while in Section~\ref{sec:axioms} we discuss how changing the axioms of a sequential product impacts the results of this paper.

\section{Preliminaries}\label{sec:prelim}
As mentioned in the introduction, our main goal is to show that a sequential product space is homogeneous and self-dual, let us start therefore with the definition of these properties.

\begin{definition}
    Let $V$ be an order unit space. An \emph{order isomorphism} is a linear map $\Phi: V\rightarrow V$ such that $\Phi(a)\geq 0 \iff a\geq 0$ for all $a\in V$. Denote the interior of the positive cone of $V$ by $C$, i.e.\ $a\in C \iff \exists \epsilon\in\R_{>0}: \epsilon 1 \leq a$. We call $V$ \emph{homogeneous} when for all $a,b \in C$ there exists an order isomorphism $\Phi$ such that $\Phi(a) = b$.
\end{definition}

\begin{definition}
    Let $V$ be an order unit space. We call $V$ \emph{self-dual} when there exists an inner product $\inn{\cdot,\cdot}$ such that for all $a\in V$: $a\geq 0$ if and only if $\inn{a,b} \geq 0$ for all $b\geq 0$.
\end{definition}

To give a complete picture of the theory of sequential product spaces we will repeat some of the known basic results regarding sequential products and sequential product spaces that can be found in for instance \cite{gudder2002sequential,gudder2005uniqueness,wetering2018characterisation}. This section on preliminaries will end with the existence proofs of spectral decompositions of effects and the corollary of homogeneity that follows from it, which was originally shown in Ref.~\cite{wetering2018characterisation} (in a slightly different setting).
    
Unless otherwise stated, we will let $V$ denote a finite-dimensional sequential product space, ${E=[0,1]_V}$ its set of effects and $\&: E\times E\rightarrow E$ a sequential product. For $a\in E$ we let $a^\perp=1-a$ denote its \emph{complement} which by virtue of $a$ lying in the unit interval of $V$ is also an effect.

\begin{proposition}\label{prop:basic}
    Let $a,b,c \in E$.
    \begin{enumerate}[label=\arabic*., ref=\arabic{counter}.\arabic*]
        \item \label{prop:unitzero} $a\mult 0 = 0\mult a = 0$ and $a\mult 1=1\mult a = a$.
        \item \label{prop:decreasing} $a\mult b \leq a$.
        \item \label{prop:orderpreserve} If $a\leq b$ then $c\mult a\leq c\mult b$.
    \end{enumerate}
\end{proposition}
\begin{proof} Originally proved in Ref. \cite{gudder2002sequential}.
    \begin{enumerate}
        \item We of course have $a\commu a$ and by \ref{ax:compadd} we have $a\commu a^\perp$. Using \ref{ax:compadd} again we then see that $a\commu a+a^\perp = 1$ so that by \ref{ax:unit} $1\mult a = a\mult 1 = a$. Using \ref{ax:compadd} again we also have $a\commu 1^\perp = 0$ so that it remains to show that $a\mult 0 = 0$. By \ref{ax:add} we get $a\mult 0 = a\mult(0+0) = a\mult 0 + a\mult 0$ so that indeed $a\mult 0 = 0$. 
        \item By the previous point and \ref{ax:add} $a = a\mult 1 = a\mult (b+(1-b)) = a\mult b + a\mult (1-b)$ so that indeed $a\mult b \leq a$, as $a\mult (1-b) \geq 0$.
        \item Since $a\leq b$ we have $b-a \geq 0$ so that using \ref{ax:add} we have $c\mult b = c\mult (b-a + a) = c\mult (b-a) + c\mult a$, from which we derive $c\mult(b-a) = c\mult b - c\mult a$. Since the lefthandside is greater than zero, the righthandside must be as well. \qedhere
    \end{enumerate}
\end{proof}

\begin{proposition} \label{prop:linearity}
    Let $a,b\in E$ and let $q$ be any rational number between zero and one, and $\lambda$ any real number between zero and one.
    \begin{enumerate}[label=\arabic*.,ref=\arabic{counter}.\arabic*]
        \item $a\mult (qb) = q(a\mult b)$.
        \item $a\mult (\lambda b) = \lambda(a\mult b)$.
        \item $(\lambda a)\mult b = a\mult (\lambda b) = \lambda(a\mult b)$.
        \item \label{prop:commumult} If $a\commu b$ then $a\commu \lambda b$.
    \end{enumerate}
\end{proposition}
\begin{proof} Originally proved in Ref. \cite{wetering2018characterisation}.
    \begin{enumerate}
        \item Of course $a\mult b = a\mult (n \frac1n b) = n a\mult (\frac1n b)$ by \ref{ax:add}. Dividing by $n$ gives $a\mult(\frac1n b) = \frac1n (a\mult b)$. By summing this equation multiple times we see that we get $a\mult (q b) = q(a\mult b)$ for any rational $0\leq q\leq 1$.

        \item Let $q_i$ be an increasing sequence of positive rational numbers that converges to $\lambda$. Using the norm of the order unit space we compute
        \begin{align*}
        \norm{\lambda (a\mult b) - a\mult (\lambda b)} &= \norm{(\lambda - q_i)(a\mult b) + q_i(a\mult b) - a\mult (\lambda b)} \\
        &= \norm{(\lambda-q_i)(a\mult b) - a\mult ((\lambda - q_i)b)}.
        \end{align*}
        Because $(\lambda - q_i)b\leq (\lambda-q_i)\norm{b}1$ and using proposition \ref{prop:orderpreserve} we have $\norm{a\mult ((\lambda - q_i)b)} \leq \norm{a}\norm{(\lambda -q_i)b} = (\lambda - q_i)\norm{a}\norm{b}$. Then $\norm{\lambda (a\mult b) - a\mult (\lambda b)} \leq 2(\lambda - q_i)\norm{a}\norm{b}$. This expression indeed goes to zero as $i$ increases so that indeed $\lambda (a\mult b) = a\mult (\lambda b)$.

        \item Clearly $\frac{1}{n}a\commu \frac{1}{n}a$ so that by \ref{ax:compadd} $\frac{1}{n}a\commu a$. In the same way we also get $qa\commu a$ and $qa^\perp \commu a^\perp$ for any rational $0\leq q\leq 1$. Using the rule $a\commu b \implies a\commu b^\perp$ from \ref{ax:compadd} we then also get $qa^\perp \commu a$ so that $a\commu (qa+qa^\perp)=q1$. Then indeed $(q1)\mult a = a\mult (q1) = q(a\mult 1) = qa$ so that also $(qa)\mult b = (a\mult (q1))\mult b = a\mult((q1)\mult b)) = a\mult qb = q(a\mult b)$. Now let $\lambda\in[0,1]$ be a real number and let $q_i$ be a sequence of rational numbers converging to $\lambda$ so that also $q_i a \rightarrow \lambda a$ and $q_i(a\mult b) \rightarrow \lambda(a\mult b)$. Then $q_i(a\mult b) = (q_ia)\mult b \rightarrow (\lambda a)\mult b$ by \ref{ax:cont}. We conclude that $(\lambda a)\mult b = \lambda(a\mult b) = a\mult(\lambda b)$.

        \item Suppose $a\commu b$, then using the previous point $a\mult (\lambda b) = \lambda (a\mult b) = \lambda (b\mult a) = (\lambda b)\mult a$. \qedhere
    \end{enumerate}
\end{proof}

As a result of this proposition, the \emph{left-multiplication map} $L_a:E\rightarrow E$ for $a\in E$ given by $L_a(b) = a\mult b$ can be extended by linearity to the entirety of $V$ by $L_a(b-c) = L_a(b) - L_a(c)$. Similarly we can define the sequential product for any element in the positive cone of $V$ (not necessarily below the identity) by rescaling: $a\mult b := \norm{a} ((\frac{1}{\norm{a}} a)\mult b)$.

\begin{definition}
    An effect $p\in E$ is called \emph{sharp} when the only effect below both $p$ and $p^\perp$ is the zero effect, i.e\ when the following implication holds: $b\leq p$ and $b\leq p^\perp$ implies $b=0$.
\end{definition}

When $V=B(H)^{\text{sa}}$ the sharp effects are precisely the projections. This should be clear considering the following proposition:

\begin{proposition}\label{prop:sharpness}
    Let $a\in E$ be an effect, $a$ is sharp if and only if $a\mult a^\perp = 0$ if and only if $a\mult a = a$.
\end{proposition}
\begin{proof} Originally proved in Ref. \cite{gudder2002sequential}.

        \noindent The equivalence of $a\mult a^\perp =0$ and $a\mult a = a$ follows straightforwardly by \ref{ax:add} and proposition \ref{prop:unitzero}: $a = a\mult 1 = a\mult (a+a^\perp) = a\mult a + a\mult a^\perp$.

        So let us assume $a$ is sharp. By \ref{ax:compadd} we have $a\commu a^\perp$ so that $a\mult a^\perp = a^\perp \mult a$. By \ref{prop:decreasing} we have $a\mult a^\perp \leq a$ and similarly $a^\perp \mult a \leq a^\perp$, the expression $a\mult a^\perp$ is therefore below both $a$ and $a^\perp$ and by sharpness of $a$ has to be zero. Now suppose $a\mult a^\perp = 0$ and let $b\leq a$ and $b\leq a^\perp$. If $b\leq a^\perp$ then by \ref{prop:decreasing} we get $a\mult b \leq a\mult a^\perp =0$, and similarly we get $a^\perp \mult b = 0$. By \ref{ax:orth} we have $b\commu a$ and $b\commu a^\perp$ so that $b = b\mult 1 = b\mult (a+a^\perp) = b\mult a + b\mult a^\perp = a\mult b + a^\perp \mult b = 0 + 0 = 0$.
\end{proof}

Let us now introduce the notion of orthogonal effects which was hinted at in \ref{ax:orth}:
\begin{definition}
    We call two effects $p$ and $q$ \emph{orthogonal} when $p\mult q = 0$.
\end{definition}
Of course by \ref{ax:orth} orthogonality is a symmetric relation, and we note that therefore orthogonal effects are also compatible.

\begin{definition}
    Let $a\in E$ be an effect. We define the powers of $a$ inductively to be $a^0 := 1$ and $a^n := a\mult a^{n-1}$. We define the \emph{classical algebra of $a$} to be the linear space $C(a)$ spanned by all the powers of $a$ and $a^\perp$.
\end{definition}

\begin{proposition}\label{prop:polyspace}
    Let $a\in E$ be an effect. $C(a)$ is a commutative sequential product space.
\end{proposition}
\begin{proof}
    $C(a)$ inherits the order structure from $V$ in the obvious way. Of course $a\commu a^\perp$ and because of \ref{ax:compmult} we have $a^n \commu (a^\perp)^m$ for all $n$ and $m$. Because of \ref{ax:compadd} and proposition \ref{prop:commumult} linear combinations of compatible effects are also compatible and hence all effects of $C(a)$ are compatible.
\end{proof}

\noindent We are now in a position to use the seminal representation theorem from Kadison:

\begin{proposition}[cf.~{\cite{kadison1951representation}}]
     Let $V$ be a complete order unit space with a bilinear operation $\circ$ that preserves positivity: $a\circ b \geq 0$ when $a,b\geq 0$. Then there exists a compact Haussdorff space $X$ such that $V\cong C(X)$, the space of continuous real-valued functions on $X$. This isomorphism is both an order- and algebra-isomorphism.
\end{proposition}

\begin{proposition}
    Let $a\in E$ be an effect. $C(a)\cong \R^n$ for some $n\in \N$.
\end{proposition}
\begin{proof}
    The sequential product is linear in the second argument. Since $C(a)$ is a commutative sequential product space by proposition \ref{prop:polyspace}, its product is also linear in the first argument, and hence this operation is bilinear. It obviously preserves positivity by definition, so that Kadisons theorem applies and $C(a) \cong C(X)$. Since $V$ is finite-dimensional, $C(a)$ has to be so as well, and hence $C(X)$ is finite-dimensional. Of course $C(X)$ is finite-dimensional only when $X$ is finite so that $X$ is necessarily discrete. We conclude that indeed $C(X) \cong \R^n$.
\end{proof}

\begin{corollary}
    Let $a\in E$ be an effect. There exists a set of orthogonal sharp effects $p_i$ compatible with $a$ and positive scalars $\lambda_i$ such that $a = \sum_i \lambda_i p_i$.
\end{corollary}
\begin{proof}
    By the previous proposition $C(a) \cong \R^n$ and this space is obviously spanned by orthogonal sharp effects, hence we can find the desired $p_i$ and $\lambda_i$. By construction $p_i\in C(a)$ so that they are compatible with $a$.
\end{proof}
\noindent We will refer to a decomposition of $a$ in the above sense as a \emph{spectral decomposition} of $a$. The existence of these decomposition is already enough to show that the space must be homogeneous:

\begin{proposition}\label{prop:homogen}
    Let $C$ denote the cone of \emph{strictly positive} elements in $V$, i.e\ the elements $v\in V$ such that $\exists \epsilon>0$ with $\epsilon 1\leq v$. The cone $C$ is \emph{homogeneous}, i.e.\ for every $v,w\in C$ there exists an order isomorphism $\Phi: V\rightarrow V$ such that $\Phi(v)=w$.
\end{proposition}
\begin{proof}
    Originally proved in Ref. \cite{wetering2018characterisation}.

    \noindent For an arbitrary positive element $a$ we can find a spectral decomposition $a=\sum_i \lambda_i p_i$ such that $\lambda_i > 0$, i.e.\ we don't write the zero `eigenvalues'. It is then straightforward to check that $a$ lies in the interior of the positive cone if and only if $\sum_i p_i=1$. In that case we define its \emph{inverse} $a^{-1}=\sum_i \lambda_i^{-1} p_i$. Since the $p_i$ are all compatible and $a$ and $a^{-1}$ are linear combinations of these effects, they are also compatible and we calculate $a\mult a^{-1} = \sum_{i,j} \lambda_i \lambda_j^{-1} p_i \mult p_j = \sum_i \lambda_i \lambda_i^{-1} p_i = \sum_i p_i = 1$ so that $a^{-1}$ is indeed the inverse of $a$ with respect to the sequential product. The multiplication map $L_a(b) := a\mult b$ is positive and has a positive inverse $L_{a^{-1}}$ due to \ref{ax:assoc}: $a^{-1}\mult (a\mult b) = (a^{-1}\mult a)\mult b = 1\mult b = b$. The map $L_a$ is therefore an order isomorphism when $a$ is strictly positive. Now, for $a$ and $b$ strictly positive and hence invertible, define $\Phi: V\rightarrow V$ by $\Phi = L_bL_{a^{-1}}$. As this is a composition of order isomorphisms, it is also an order isomorphism and of course $\Phi(a) = b\mult (a^{-1}\mult a) = b\mult 1 = b$ as desired.
\end{proof}

\section{Proof of self-duality}\label{sec:selfdual}
With homogeneity of $V$ now established, we set our sights on proving self-duality. We do this in a few steps. First we study the lattice of sharp effects in Section~\ref{sec:subsharpeffects}. We then consider properties of the \emph{atoms} of this lattice in Section~\ref{sec:atomeffect}. Then in Section~\ref{sec:coverprop} we establish that this lattice has the \emph{covering property} as defined in Ref.~\cite{alfsen2012state}. The covering property states that for every sharp effect $p$ there is a unique number $r$ called the \emph{rank} of $p$ such that we can write $p=\sum_{i=1}^r p_i$ where the $p_i$ are atomic and orthogonal. Using this definition we can define the rank of a space as equal to the rank of the unit effect. The existence of well-defined ranks of sharp effects allows us to reduce the question of self-duality to that of self-duality in rank 2 spaces. This problem is in turn solved by appealing to the classification result of Ref.\ \cite{ito2017p} that homogeneous spaces of rank 2 are always self-dual, which is done in Section~\ref{sec:subselfdual}.

\subsection{The lattice of sharp effects}\label{sec:subsharpeffects}

\begin{proposition}\label{prop:sharpprop} Let $a\in E$ be any effect and let $p\in E$ be sharp.
    \begin{enumerate}[label=\arabic*., ref=\arabic{counter}.\arabic*]
        \item \label{prop:belowsharp} $a\leq p$ if and only if $p\mult a = a\mult p = a$ if and only if $p^\perp \mult a = 0$.
        \item \label{prop:abovesharp} $p\leq a$ if and only if $p\mult a = a\mult p = p$.
    \end{enumerate}
\end{proposition}
\begin{proof}~
    \begin{enumerate}
        \item Suppose $a\leq p$ with $p$ sharp. Then $p^\perp \mult a \leq p^\perp \mult p = 0$ by proposition \ref{prop:sharpness} and \ref{prop:decreasing}. Hence $a\commu p^\perp$ and $a\commu p$ so that $a = a\mult (p+p^\perp) = a\mult p + a\mult p^\perp = a\mult p = p\mult a$. For the other direction we note that $a=p\mult a \leq p$ by \ref{prop:decreasing}.

        \item Suppose $p\leq a$ with $p$ sharp, then $a^\perp \leq p^\perp$ with $p^\perp$ sharp so that by the previous point $a\commu p$ and $p\mult a^\perp = 0$ so that $p=p\mult(a+a^\perp) = p\mult a$. \qedhere
    \end{enumerate}
\end{proof}

\begin{definition}
    For an effect $a\in E$ we let $\ceil{a}$ denote the smallest sharp element above $a$ (when it exists) and $\floor{a}$ the largest sharp element below $a$ (when it exists).
\end{definition}

\begin{proposition}
    The ceiling and the floor exist for any $a$. Moreover, writing $a=\sum_i \lambda_i p_i$ with $1\geq \lambda_i> 0$ and the $p_i$ sharp and orthogonal, then $\ceil{a} = \sum_i p_i$ and $\floor{a}=\ceil{a^\perp}^\perp$.
\end{proposition}
\begin{proof}
    Write $a=\sum_i \lambda_i p_i$. Of course $\sum_i p_i$ is an upper bound of $a$. Suppose $a\leq r$ for some sharp $r$. Then $\lambda_i p_i \leq r$, so by proposition \ref{prop:belowsharp} $r\mult (\lambda_i p_i) = \lambda_i p_i$. Using linearity we can rewrite this expression to $\lambda_i (r\mult p_i) = \lambda_i p_i$ so that $r\mult p_i = p_i$. So again by \ref{prop:belowsharp} $p_i\leq r$ and $p_i\commu r$ from which we get $r\mult \sum_i p_i = \sum_i r\mult p_i = \sum_i p_i$ so that also $\sum_i p_i\leq r$ which proves that it is the least upper bound. The other statement now follows because $a\leq b \iff b^\perp \leq a^\perp$.
\end{proof}

As a corollary of the above we also see that $\ceil{\lambda a} = \ceil{a}$ when $1\geq \lambda > 0$ and that $a$ is sharp if and only if $\ceil{a}=a$ or $\floor{a}=a$. We also note that $a\leq b$ implies that $\ceil{a}\leq \ceil{b}$.

\begin{proposition}\label{prop:lattice}
    The sharp effects form a lattice: for two sharp effects $q$ and $p$, their least upper bound $p\vee q$ and greatest lower bound $p\wedge q$ exist. Furthermore the following relation holds between them: $(p\vee q)^\perp = p^\perp \wedge q^\perp$.
\end{proposition}
\begin{proof}
    We claim that $p\vee q = \ceil{\frac12(p+q)}$. Note that $p\leq p+q$ and thus that $\frac12 p \leq \frac12(p+q)$ so that $p = \ceil{p} = \ceil{\frac12p} \leq \ceil{\frac12(p+q)}$. Similarly we also have $q\leq \ceil{\frac12(p+q)}$ and thus this is an upper bound. Suppose now that $p\leq a$ and $q\leq a$ for some sharp $a$. Then of course also $\frac12(p+q)\leq \frac12(a+a) = a$. Taking the ceiling on both sides then shows that indeed $p\vee q = \ceil{\frac12(p+q)}$.

    To find $p\wedge q$ we note that $(\cdot)^\perp$ is an order-antiautomorphism, and thus that it interchanges joins with meets: $(p\vee q)^\perp = p^\perp\wedge q^\perp$.
\end{proof}

\begin{proposition}\label{prop:sharplattice} Let $a\in E$ be any effect and let $p\in E$ be sharp.
    \begin{enumerate}[label=\arabic*., ref=\arabic{counter}.\arabic*]
        \item \label{prop:meetsharp} $p\mult a = 0$ if and only if $p+a\leq 1$ in which case it is the least upper bound of the two. When $p\mult a = 0$ their sum $p+a$ is sharp if and only if $a$ is also sharp.
        \item \label{prop:joinsharp} If both $a$ and $p$ are sharp and $a$ and $p$ are compatible then $p\mult a$ is sharp and equal to their join: $p\wedge a = p\mult a$.
    \end{enumerate}
\end{proposition}
\begin{proof}Originally proved in Ref. \cite{gudder2002sequential}.
    \begin{enumerate}
        \item $p\mult a = 0$ if and only if $p^\perp\mult a = a$ which by proposition \ref{prop:abovesharp} is true if and only if $a\leq p^\perp = 1-p$ so that indeed $p+a\leq 1$. That $p+a$ is an upper bound of $p$ and $a$ is obvious. Suppose now that $b$ is also an upper bound so that $p\leq b$ and $a\leq b$. 
        We then calculate using proposition \ref{prop:abovesharp} $p = p\mult b = p\mult(b-a + a) = p\mult(b-a) + p\mult a = p\mult(b-a)$ so that $p\leq b-a$ again by \ref{prop:abovesharp}. This gives $p+a \leq b$ so that $p+a$ is indeed the least upper bound. 

        Now to show $p+a$ is sharp if and only if bot $p$ and $a$ are sharp: since $p\mult a = 0$ we have $p\commu a$ and thus also $p\commu p+a$ and $a\commu p+a$ by \ref{ax:compadd}. We calculate $(p+a)\mult (p+a) = p\mult p + 2 p\mult a + a\mult a = p + a\mult a = (p+a) + (a-a\mult a)$. We therefore have $(p+a)\mult (p+a) = p+a$ if and only if $a-a\mult a = 0$ which proves the result by proposition \ref{prop:sharpness}.

        \item Suppose both $p$ and $a$ are sharp and that $p\commu a$. We calculate:
         $$(p\mult a)\mult (p\mult a) = (p\mult a)\mult(a\mult p) =p\mult(a\mult(a\mult p)) = p\mult (a\mult p) = p\mult (p\mult a) = p\mult a$$
        where we have used that $a\commu a\mult p$ and $p\commu a\mult p$ by \ref{ax:compmult}. Hence $p\mult a$ is sharp. It is a lower bound of $p$ and $a$ by \ref{prop:decreasing}. Suppose $b\leq p, a$ is also a lower bound. We calculate $p\mult a = p\mult (a-b + b) = p\mult (a-b) + p\mult b = p\mult(a-b) + b \geq b$, where we have used that $p\mult b = b$ as a consequence of proposition \ref{prop:belowsharp}.\qedhere
    \end{enumerate}
\end{proof}

\begin{lemma} \label{lem:ceilzero}
    Let $a,b\in E$. If $b\mult a = 0$ then $b\mult \ceil{a} = 0$.
\end{lemma}
\begin{proof}
    Write $a= \sum_i \lambda_i p_i$. If $b\mult a = 0 = \sum_i \lambda_i b\mult p_i$, then we must have $b\mult p_i=0$ for all $p_i$. Since $\ceil{a}=\sum_i p_i$ the claim follows.
\end{proof}

\begin{lemma}\label{lem:ceilceil}
    Let $p\in E$ be sharp and $a\in E$ arbitrary, then $\ceil{p\mult a} = \ceil{p\mult \ceil{a}}$.
\end{lemma}
\begin{proof}
    First of all $p\mult a \leq p\mult \ceil{a}$ so that $\ceil{p\mult a}\leq \ceil{p\mult \ceil{a}}$. It suffices therefore to prove the other inequality. Because $p\mult a \leq p\mult 1 = p$ we also have $\ceil{p\mult a}\leq \ceil{p}=p$ implying $\ceil{p\mult a}^\perp\mult p =0$ so that $\ceil{p\mult a}^\perp$ and $p$ are compatible. 
    Now because $p\mult a \leq \ceil{p\mult a}$ we can use proposition \ref{prop:belowsharp} to write $0=\ceil{p\mult a}^\perp \mult (p\mult a) = (\ceil{p\mult a}^\perp\mult p)\mult a = (\ceil{p\mult a}^\perp \mult p)\mult \ceil{a} = \ceil{p\mult a}^\perp\mult (p\mult \ceil{a})$ where we have used lemma \ref{lem:ceilzero} to replace $a$ with $\ceil{a}$. Since $\ceil{p\mult a}^\perp\mult (p\mult \ceil{a}) = 0$ we use \ref{prop:belowsharp} again to conclude $p\mult \ceil{a}\leq \ceil{p\mult a}$ so that indeed $\ceil{p\mult \ceil{a}}\leq \ceil{p\mult a}$.
\end{proof}

\subsection{Atomic effects}\label{sec:atomeffect}

\begin{definition}
    An effect $p\in E$ is \emph{atomic} when $p$ is a nonzero sharp effect and if for all $a\in E$ with $a\leq p$ we have $a=\lambda p$ for some $\lambda\in[0,1]$.
\end{definition}

The name atomic comes from the fact that in the lattice of sharp effects, the atomic effects are the smallest nonzero elements. It turns out that the lattice of effects in a finite-dimensional sequential product space is \emph{atomic} meaning that any sharp effect can be written as a sum of atomic effects:

\begin{proposition}
    Every sharp effect can be written as a sum of orthogonal atomic effects.
\end{proposition}
\begin{proof}
    Let $p$ be sharp and suppose it is not atomic, then we can find $0\leq a\leq p$ such that $a\neq \lambda p$ for any $\lambda\in [0,1]$. Write $a= \sum_i \lambda_i q_i$ where the $q_i$ are sharp and orthogonal. Then $\lambda_i q_i\leq p$ and thus also $\ceil{\lambda_i q_i} = q_i\leq \ceil{p}=p$. If all the $q_i$ are equal to $p$, then $a$ is a multiple of $p$, so at least one of the $q_i$ is strictly smaller than $p$. We can repeat this process for $q_i$ and $p-q_i$, getting a sequence of nonzero orthogonal sharp effects that sum up to $p$. Since the space is finite-dimensional and orthogonal effects are linearly independent this process must stop after a finite amount of steps in which case we are left with atomic effects.
\end{proof}
\begin{corollary}\label{cor:spectralatomic}
    Every $a\in V$ can be written as $a=\sum_i \lambda_i p_i$ where the $p_i$ are orthogonal sharp atomic effects.
\end{corollary}
\begin{proof}
    For every $a\in V$ we can find a spectral decomposition in terms of orthogonal sharp effects. The previous proposition shows that these sharp effects can be further decomposed into atomic effects.
\end{proof}

\noindent Recall that the norm in an order unit space is defined as ${\norm{a}:=\inf\{r~;~-r 1 \leq a \leq r1\}}$.

\begin{lemma} \label{lem:atomicnorm}
    A non-zero effect $p$ is atomic if and only if we have $p\mult a = \norm{p\mult a}p$ for all $a\in E$.
\end{lemma}
\begin{proof}
    First we establish that the norm of any non-zero sharp effect is equal to $1$. Let $q$ be sharp. We see that $q = q\mult q \leq q \mult(\norm{q} 1) = \norm{q} q\mult 1 = \norm{q} q$ so that $\norm{q}\geq 1$. But since $q\leq 1$ we also have $\norm{q}\leq 1$.
    
    Suppose $p$ is atomic. Because $0\leq p\mult a \leq p$ we must have $p\mult a = \lambda p$ for some $0\leq \lambda \leq 1$ so that $\norm{p\mult a} = \lambda \norm{p} = \lambda$ because $p$ is sharp.

    For the other direction first note that $p= p\mult \ceil{p} = \norm{p\mult \ceil{p}} p = \norm{p}p$ so that necessarily $\norm{p}=1$ (since $p\neq 0$). By writing $p$ as a linear combination of sharp effects we see that then also $\norm{p^2}=1$. Now $p\mult p = \norm{p^2} p = p$ so that $p$ is sharp. Let $q\leq p$ be non-zero sharp. Then $q=p\mult q = \norm{p\mult q}p=p$ (using again that $\norm{q}=1$ because $q$ is sharp) so there are no non-trivial sharp effects below $p$. Now if $a=\sum_i \lambda_i q_i$ lies below $p$ we see that $\lambda_iq_i\leq p$ so that $\ceil{\lambda_i q_i} = q_i\leq \ceil{p}=p$ so that $q_i=p$ and thus $a=\lambda p$. Since all $a\in E$ can be written in this way we conclude that this holds for all $a\leq p$, so that $p$ is indeed atomic.
\end{proof}
\begin{corollary}
    The set of atomic effects is closed in the norm topology.
\end{corollary}
\begin{proof}
    Let $p_n\rightarrow p$ be a norm converging set of atomic effects $p_n$. We need to show that $p$ is also atomic. As a result of the previous lemma we have $p_n\mult a = \norm{p_n \mult a} p_n$ for all effects $a$. By continuity of $\&$ (i.e.\ axiom \ref{ax:cont}) we have $p_n\mult a \rightarrow p\mult a$ so that $p\mult a = \lim p_n \mult a = \lim \norm{p_n\mult a}p_n = \norm{p\mult a} p$. Using the previous lemma again we conclude that $p$ is indeed atomic.
\end{proof}

\begin{proposition}\label{prop:atompreservation}
    Let $a\in E$ be arbitrary and $p\in E$ be atomic, then $a\mult p$ is proportional to an atomic effect.
\end{proposition}
\begin{proof}
    The property that $0\leq a\leq p \implies a=\lambda p$ is determined by the order, so any order isomorphism preserves it. If $a$ is invertible then $L_a: V\rightarrow V$ given by $L_a(b):= a\mult b$ is an order isomorphism, so that $L_a(p)$ must be proportional to an atomic effect. For non-invertible $a$ we write $a_n = a+\frac{1}{n}$, so that $a_n$ is invertible and the sequence $a_n$ converges to $a$. Let $q_n = (a_n\mult p)/\norm{a_n \mult p}$, then all the $q_n$ are atomic. By the continuity condition \ref{ax:cont} we get $a_n\mult p \rightarrow a\mult p$ so that also $\norm{a_n \mult p}\rightarrow \norm{a\mult p}$. The sequence $q_n$ is therefore also convergent and since the set of atomic effects is closed by the previous corollary we conclude that $q_n\rightarrow q=(a\mult p)/\norm{a\mult p}$ is atomic.
\end{proof}

\subsection{The Covering Property}\label{sec:coverprop}

At this point we know that the set of sharp effects forms an atomic lattice, but in fact it has the much stronger \emph{covering property} that allows us to attach a \emph{rank} to each sharp effect: the amount of atomic effects needed to make the effect. To show this we need some results from Alfsen and Shultz \cite{alfsen2012geometry} that unfortunately were proven in a slightly different setting. We will repeat these results with very similar proofs, but adapted to work in the setting of sequential product spaces.

\begin{lemma}[{cf.~\cite[Lemma 8.9]{alfsen2012geometry}}]\label{lem:wedgeminus}
     Let $q$ and $p$ be sharp with $q\leq p$, then $p-q = p\wedge q^\perp$.
\end{lemma}
\begin{proof}
    Let $q$ and $p$ be sharp with $q\leq p$, then $p-q$ is sharp and $p\commu q$, so also $p\commu q^\perp$ so that $p-q = p\mult q^\perp = p\wedge q^\perp$ by Proposition~\ref{prop:joinsharp}.
\end{proof}

\begin{lemma}[{cf.~\cite[Theorem 8.32]{alfsen2012geometry}}] Let $p$ be sharp and $a$ arbitrary, then $\ceil{p\mult a} = (\ceil{a}\vee p^\perp)\wedge p$.
\end{lemma}
\begin{proof}
    Because $\ceil{p\mult a} = \ceil{p\mult \ceil{a}}$ by Lemma~\ref{lem:ceilceil} it suffices to prove this for sharp $a$. We prove the equality by showing that an inequality holds in both directions.

    Since $p^\perp \leq a\vee p^\perp$ we have $p^\perp \commu (a\vee p^\perp)$ by Proposition~\ref{prop:belowsharp} so that in turn $p\commu (a\vee p^\perp)$ by \ref{ax:compadd}.
    We proceed by using \ref{ax:assoc}: $(a\vee p^\perp)\mult (p\mult a) = ((a\vee p^\perp)\mult p)\mult a = p\mult ((a\vee p^\perp)\mult a) = p\mult a$ because $a\vee p^\perp \geq a$. Therefore $p\mult a \leq a\vee p^\perp$ which implies that $\ceil{p\mult a} \leq a\vee p^\perp$. Since also $p\mult a\leq p$ and therefore $\ceil{p\mult a}\leq p$ we conclude that $\ceil{p\mult a}\leq (a\vee p^\perp)\wedge p = (\ceil{a}\vee p^\perp)\wedge p$.

    Now for the other direction: we obviously have $p^\perp\mult (p\mult a) = (p^\perp \mult p)\mult a = 0$ by \ref{ax:compadd} and \ref{ax:assoc} so that by Lemma~\ref{lem:ceilzero} $p^\perp\mult \ceil{p\mult a}=0$. Using Proposition~\ref{prop:belowsharp} we see then that $p\commu \ceil{p\mult a}^\perp$ and therefore by Proposition~\ref{prop:joinsharp} that $\ceil{p\mult a}^\perp\mult p = \ceil{p\mult a}^\perp\wedge p$. Since $p\mult a \leq \ceil{p\mult a}$ we calculate using Proposition~\ref{prop:belowsharp}: $0=\ceil{p\mult a}^\perp\mult (p\mult a) = (\ceil{p\mult a}^\perp \mult p)\mult a = (\ceil{p\mult a}^\perp \wedge p)\mult a$ so that $a\leq (\ceil{p\mult a}^\perp\wedge p)^\perp = \ceil{p\mult a}\vee p^\perp$ by Proposition~\ref{prop:lattice}. 
    Then of course also $a\vee p^\perp \leq \ceil{p\mult a}\vee p^\perp$ and by noting that $\ceil{p\mult a}$ and $p^\perp$ are orthogonal and using Proposition~\ref{prop:meetsharp}: $\ceil{p\mult a}\vee p^\perp = \ceil{p\mult a}+p^\perp$. Bringing the $p^\perp$ to the other side and using Lemma~\ref{lem:wedgeminus} (which applies because $p^\perp\leq a\vee p^\perp$) then gives $(a\vee p^\perp)\wedge p = a\vee p^\perp - p^\perp \leq \ceil{p\mult a}$.
\end{proof}

\begin{proposition}[{cf.~\cite[Proposition 9.7]{alfsen2012geometry}}]\label{prop:coverprop}
     The lattice of sharp effects has the \emph{covering property}: for $q$ atomic, the expression $(q\vee p)\wedge p^\perp = (q\vee p)-p$ is either zero or atomic. In other words: when $q$ does not lie below $p$ then there is no sharp effect lying strictly between $p$ and $q\vee p$.
\end{proposition}
\begin{proof}
    By the previous lemma $(q\vee p)\wedge p^\perp = \ceil{p^\perp \mult q}$. Since $p^\perp\mult q$ is proportional to an atom by Proposition~\ref{prop:atompreservation}, it is either zero (when $q\leq p$) in which case we are done, or non-zero in which case $\ceil{p^\perp \mult u}$ is an atom, which also proves the statement. The equality $(q\vee p)\wedge p^\perp = (q\vee p)-p$ follows directly from Lemma~\ref{lem:wedgeminus}.

    The last observation is proven as follows. Suppose $p< r < q\vee p$. Subtract $p$ to get $0< r-p < q\vee p - p$.  As $r-p$ is sharp and $q\vee p -p$ has been established to be atomic, this is not possible.
\end{proof}

\begin{definition}
    Let $p$ be sharp and let $p_i$ be a collection of atomic orthogonal effects such that $p=\sum_i^n p_i$. The minimal size of such a collection is called the \emph{rank} of $p$. The rank of a sequential product space is defined to be the rank of the unit effect.
\end{definition}

With the covering property proven we can finally prove the following `dimension' theorem:

\begin{proposition}[{cf.~\cite[Proposition 1.66]{alfsen2012state}}] Write $p=\sum_i^n p_i$ where the $p_i$ are orthogonal and atomic, then $n=\rnk~ p$, i.e.\ all ways of writing $p$ as a sum of atomic effects require an equal amount of atomic effects. Furthermore, suppose $q\leq p$ then $\rnk~ q \leq \rnk~ p$ and if also $\rnk~ q = \rnk~ p$ then necessarily $q=p$.
\end{proposition}
\begin{proof}
    Let $p^\prime = p_1\vee\ldots\vee p_{n-1}$. Then $p^\prime\vee p_n = p$ and by the covering property (proposition \ref{prop:coverprop}) there is no sharp effect strictly between $p^\prime$ and $p$. Suppose now $q\leq p$ is atomic and suppose that $q$ is not below $p^\prime$. Then $p^\prime \vee q$ must be strictly greater than $p^\prime$, but since this must also lie below $p$ we conclude that $p^\prime \vee q = p$.

    Let $p = \sum_j^r q_j = q_1\vee\ldots\vee q_r$ where $r:=\rnk ~ p$ is the minimal amount of terms needed to write $p$ as a sum of atomic effects. We must then of course have $r\leq n$. Let $q=q_2\vee\ldots\vee q_r$, then $q$ must lie strictly below $p$ since $q_1\leq p$ but not $q_1\leq q$. It then follows that there must be a $p_i$ such that $p_i$ does not lie below $q$ as well, since otherwise $p = p_1\vee\ldots\vee p_n \leq q < p$. Without loss of generality let this $p_i$ be $p_1$. By the previous paragraph we must have $p_1\vee q = p_1\vee q_2\ldots\vee q_r = p$. This procedure can be repeated with $q_2,\ldots, q_r$ until we are left with the equation $p_1\vee \ldots\vee p_r = p$. Suppose $n>r$, then because $p_n$ is orthogonal to all the other $p_i$'s we have in particular $p_n\leq p_1^\perp \wedge \ldots \wedge p_r^\perp = (p_1\vee \ldots \vee p_r)^\perp = p^\perp$. Since also $p_n\leq p$ we get $p_n=0$ by sharpness which is a contradiction. We therefore have $n=r$.

    Now suppose $q=\sum_j^s q_j\leq p=\sum_i^r p_i$. Where $s=\rnk~ q$. Since $p-q$ is sharp we can write $p-q = \sum_k^t v_k$. Then because $p = \sum_j^s q_j + \sum_k^t v_k$ we must by the previous points have $s+t = r$ so that indeed $\rnk~ q\leq \rnk~ p$. When these ranks are equal we must have $t=0$ so that indeed $p-q = 0$.
\end{proof}
\begin{corollary}\label{cor:rank}
    Let $p\neq q$ be two atomic sharp effects and suppose $0\leq a\leq p\vee q$, then $a=\lambda_1 r_1 + \lambda_2 r_2$ where the $r_i$ are orthogonal and atomic and $r_1 + r_2 = p\vee q$.
\end{corollary}
\begin{proof}
    By proposition \ref{prop:coverprop} $(p\vee q) - p$ is atomic so that $p\vee q$ can be written as the sum of two atomic sharp effects so that indeed $\rnk~ p\vee q = 2$. Suppose $0\leq a \leq p\vee q$. Let $a=\sum_i^n \lambda_i r_i$ be a spectral decomposition of $a$ with the $r_i$ orthogonal and atomic. Of course $\ceil{a} \leq p\vee q$ so that by the previous proposition we must have $\rnk \ceil{a} \leq 2$. Since also by the previous proposition $\rnk~\sum_i^n r_i = n$ we see that we must have $n=2$ and thus that $a$ is as desired.
\end{proof}

\subsection{Self-duality}\label{sec:subselfdual}
The important concept of this section will be that of \emph{strict convexity} of a cone, since this is related to a characterisation theorem for homogeneous spaces.
\begin{definition}
    Let $C$ be a positive cone of an order unit space $V$. We call $F\subseteq C$ a \emph{face} of $C$ if $F$ is a convex set such that whenever $\lambda a+(1-\lambda) b \in F$ with $0<\lambda<1$ then $a,b \in F$. The face $\{\lambda p~;~ \lambda \in \R_{\geq 0}\}$ of $C$ defined by an extreme point $p\in C$ is called an \emph{extreme ray}. A face is called \emph{proper} when it is non-empty and not equal to $C$. If the only proper faces of a cone are extreme rays the cone is \emph{strictly convex}. 
\end{definition}
\begin{proposition}[{cf.~\cite{ito2017p}}]\label{prop:itochar}
    Let $V$ be a finite-dimensional ordered vector space with a strictly convex homogeneous positive cone, then $V$ is order isomorphic to a spin-factor, i.e.\ $V\cong H\oplus \R$ where $H$ is a real finite-dimensional Hilbert space with the order on $H\oplus \R$ given by $(v,t)\geq 0 \iff t\geq \norm{v}_2$.
\end{proposition}

\begin{definition}\label{def:orderideal}
    Let $p$ and $q$ be two unequal atomic effects. We define the \emph{order ideal} generated by $p$ and $q$ as $V_{p,q}:= \{v\in V~;~\exists n: -n~p\vee q \leq v \leq n ~p\vee q\}$.
\end{definition}
$V_{p,q}$ is an order unit space with unit $p\vee q$. If we have $a,b\in [0,1]_{V_{p,q}}$ then $a\mult b \leq a \leq p\vee q$ so that the sequential product restricts to this space. We therefore conclude by Proposition~\ref{prop:homogen} that this space has a homogeneous positive cone. 

\begin{lemma}
    Let $p$ and $q$ be two unequal atomic effects. The positive cone of $V_{p,q}$ is strictly convex.
\end{lemma}
\begin{proof}
    Let $F$ be a proper face of the positive cone of $V_{p,q}$. Let $a\in F$ and write $a=\lambda (\lambda^{-1} a) + \lambda^\perp 0$, so that $\lambda^{-1} a \in F$. We see that $F$ is closed under positive scalar multiplication and thus that we can restrict ourselves to effects. Let $a\in F$ be an effect. By Corollary~\ref{cor:rank} we can write $a=\lambda r + \mu r^\perp$ for some $\lambda,\mu\geq 0$ and $r$ atomic. Suppose both $\lambda,\mu >0$, then because $F$ is a face $r,r^\perp \in F$ so that $\frac{1}{2}1 = \frac{1}{2}(r+r^\perp)\in F$. But then since $1=s+s^\perp$ for any atomic $s$ we see that $F$ has to be the entire positive cone. We conclude that we must have had $a=\lambda r$ for some atomic $r$. If there were some other atomic $s \in F$, then we can consider $a = \frac{1}{2}(r + s)$. We know that $a$ can't be atomic so we can write it as $a=\lambda r + \mu r^\perp$ with $\lambda, \mu > 0$ which is a contradiction. We conclude that $F$ is an extreme ray and thus that the positive cone of $V_{p,q}$ is strictly convex.
\end{proof}

\begin{corollary}
    Let $p$ and $q$ be two different atomic effects, then $V_{p,q}$ is isomorphic to a spin-factor.
\end{corollary}
\begin{proof}
    Follows directly from the previous lemma and Proposition~\ref{prop:itochar}.
\end{proof}

Recall that a \emph{state} on an order unit space is a positive linear map $\omega: V\rightarrow \R$ such that $\omega(1) = 1$. For an atomic effect $p$ in a spin-factor there exists a unique state $\omega_p$ such that $\omega_p(p)=1$. A spin-factor has \emph{symmetry of transition probabilities} \cite{alfsen2012geometry}: $\omega_p(q)=\omega_q(p)$ for any two atomic effects $p$ and $q$. We can use the previous results to prove that symmetry of transition probabilities also holds for arbitrary (finite-dimensional) sequential product spaces.

\begin{proposition}\label{prop:uniquestate}
     Let $p,q \in E$ be atomic effects. There exist unique pure states $\omega_p$ and $\omega_q$ such that $\omega_p(p)=1$ and $\omega_q(q)=1$. Furthermore for these states we have $\omega_p(q)=\omega_q(p)$.
\end{proposition}
\begin{proof}
    The states separate the points of an order unit space \cite[Corollary 1.27]{alfsen2012state} so that for $p$ we can find a state $\omega$ such that $\omega(p)\neq 0$. Define $\omega_p(a) := \omega(p\mult a)/(\omega(p))$, then $\omega_p$ is a state and $\omega_p(p)=1$. Suppose there is another state $\omega^\prime$ such that $\omega^\prime(p)=1$. Let $q\neq p$ be any other atomic effect (if there is no atomic $q\neq p$ then $V\cong \R$ and we are already done) and look at the restrictions of the states $\omega_p$ and $\omega^\prime$ to the space $V_{p,q}$. These restriction maps are still states as $\omega_p(p\vee q)\geq \omega_p(p)=1$ (and similarly for $\omega^\prime$). Because states with this property are unique on spin-factors we can conclude that these restricted states are equal on this subspace and in particular $\omega_p(q)=\omega^\prime(q)$. Since $q$ was arbitrary and the atomic effects span $V$ we conclude that $\omega_p=\omega^\prime$ so that $\omega_p$ is indeed unique.

    For any two atomic $p$ and $q$ we can look at their unique pure states $\omega_p$ and $\omega_q$ as restricted to $V_{p,q}$ for which we know that $\omega_p(q)=\omega_q(p)$ which finishes the proof.
\end{proof}

\begin{proposition}\label{prop:atomicprod}
    Let $p$ and $q$ be atomic sharp effects. $p$ and $q$ are orthogonal, i.e.\ $p\mult q = q\mult p = 0$ if and only if $\omega_p(q) = \omega_q(p) = 0$. Furthermore, $p\mult q = \omega_p(q) p$.
\end{proposition}
\begin{proof}
    Note that if $q\mult p =0$ then by proposition \ref{prop:meetsharp} $q+p\leq 1$ so that $1=\omega(p)\leq\omega_p(q+p)\leq \omega(1)=1$ from which we conclude that $\omega_p(q)=0$. So if $p$ and $q$ are orthogonal then $\omega_p(q)=\omega_q(p)=0$.

    For the converse we will show that $p\mult q = \omega_p(q) p$, from which it directly follows that $\omega_p(q) = 0 \implies p\mult q = 0$. Since $p$ is atomic we of course have $p\mult q = \lambda p$ for some $\lambda \geq 0$. Let $\omega^\prime(a) = \omega_p(p\mult a)$, then $\omega^\prime(p) = \omega_p(p\mult p) = \omega_p(p) = 1$, so that by the uniqueness of $\omega_p$ we have $\omega^\prime = \omega_p$. We then see that $\omega_p(q) = \omega^\prime(q) = \omega_p(p\mult q) = \omega_p(\lambda p) = \lambda \omega_p(p) = \lambda$.
\end{proof}

\begin{proposition}
    There exists an inner product on  $V$ such that the positive cone is self-dual with respect to this inner product.
\end{proposition}
\begin{proof}
    For atomic $p$ and $q$ we let $\inn{p,q}:= \omega_p(q)=\omega_q(p)=\inn{q,p}$. We can then extend it by linearity to arbitrary $a=\sum_i \lambda_i p_i$ and $b=\sum_j \mu_j q_j$ by $\inn{a,b}:= \sum_{i,j} \lambda_i\mu_j \inn{p_i,q_j}$. This is well-defined since $\inn{a,b} = \sum_i \lambda_i \omega_{p_i}( \sum_j \mu_j q_j) = \sum_i \lambda_i \omega_{p_i}(b) = \sum_j \mu_j \omega_{q_j}(a)$ so that this is independent of the representation of $a$ and $b$ in as linear combinations of atomic effects. Now $\inn{a,a} = \sum_{i,j} \lambda_i \lambda_j \omega_{p_i}(p_j) = \sum_i \lambda_i^2$ since $p_i$ and $p_j$ are orthogonal when $i\neq j$ and $\omega_{p_i}(p_i)=1$. We conclude that $\inn{a,a}\geq 0$ and that it is only equal to zero when $a=0$ so that $\inn{\cdot,\cdot}$ indeed is an inner product.

    If $a$ and $b$ are positive elements then we can write them as $a=\sum_i \lambda_i p_i$ and $b=\sum_j \mu_j q_j$ where all the $\lambda_i$ and $\mu_j$ are greater than zero. It then easily follows that $\inn{a,b}\geq 0$ because $\omega_{p_i}(q_j)\geq 0$. Conversely if we have $a=\sum_i \lambda_i p_i$ with $\lambda_i$ not necessarily positive with $\inn{a,b}\geq 0$ for all $b\geq 0$, then we can in particular take $b=p_j$ to see that $0\leq \inn{a,p_j} = \lambda_j$ from which we conclude that indeed $a\geq 0$.
\end{proof}

\begin{theorem}\label{theor:seqprodkoechervinberg}
    A finite-dimensional sequential product space is isomorphic to a Euclidean Jordan algebra.
\end{theorem}
\begin{proof}
    By proposition \ref{prop:homogen} the space is homogeneous, and by the previous proposition it is self-dual. The Koecher-Vinberg theorem \cite{koecher1957positivitatsbereiche} states that any finite-dimensional homogeneous self-dual ordered vector space is order-isomorphic to a Euclidean Jordan algebra.
\end{proof}

\section{The Jordan product from a sequential product}\label{sec:jordanproduct}
The Koecher-Vinberg theorem is a rather indirect way of establishing the Jordan algebra structure of the space. Since we don't have just a homogeneous self-dual space, but we also have access to the sequential product we can in fact construct the Jordan product directly. That is what we will strive for in this section. We will use the construction of the Jordan product from the work of Alfsen and Shultz \cite{alfsen2012geometry}, but then adapted to our setting. In this section we will again let $V$ be a finite-dimensional sequential product space, and hence by the previous section it is self-dual and homogeneous. We write $L_a: V\rightarrow V$ for the linear map that sends $v\in V$ to $a\mult v$.

\begin{definition}
    Let $p$ be an atomic sharp effect and let $b\in V$ be arbitrary. We define their \emph{Jordan product} as $p*b = \frac{1}{2}(\id + L_p - L_{p^\perp}) b$.
\end{definition}

\begin{lemma}\label{lem:restrict}
    Let $p$ and $q$ be atomic sharp effects, then $p^\perp \mult q = p^\prime \mult q$ where $p^\prime = p\vee q - p$.
\end{lemma}
\begin{proof}
    First note that $p^\perp = 1-p = 1-p\vee q + p\vee q -p = (p\vee q)^\perp + p^\prime$ and hence that $p^\prime \leq p^\perp$ so that $p^\prime \commu p^\perp$ by proposition \ref{prop:belowsharp}. We then also have $p^\perp\mult (p\vee q) = (p\vee q)\mult p^\perp = (p\vee q) \mult ((p\vee q)^\perp + p^\prime) = p^\prime$. Now using the fact that we are working with compatible effects and that $q\leq p\vee q$ we calculate $p^\perp \mult q = p^\perp \mult ((p\vee q)\mult q) = (p^\perp \mult (p\vee q))\mult q = p^\prime \mult q$.
\end{proof}

\begin{lemma}[{cf.~\cite[Lemma 9.29]{alfsen2012geometry}}]\label{lem:atomiccommute}
Let $p$ and $q$ be sharp atomic effects.
    \begin{enumerate}
        \item $p*q = q*p$.
        \item When $p\mult q = 0$ we have $p*q=0$ and in that case for any $b\in V$: $p*(q*b)=q*(p*b)$.
        \item $p*p = p$.
    \end{enumerate}
\end{lemma}
\begin{proof}~
    \begin{enumerate}
        \item If $p=q$ this is trivial, so assume that $p\neq q$. Let us denote $p^\prime = p\vee q - p$. By proposition~\ref{prop:coverprop} $p^\prime$ is atomic. By proposition~\ref{prop:atomicprod} we have $p\mult q = \omega_p(q) p = \inn{p,q} p$ and similarly $p^\prime \mult q = \inn{p^\prime ,q} p^\prime$.
        Expanding the definition of $p*q$ and using lemma \ref{lem:restrict} to write $p^\perp \mult q = p^\prime \mult q$ where $p^\prime = p\vee q - p$ we calculate
        \begin{align*}
        2 (p*q) &= q + \inn{p,q}p - \inn{p^\prime, q} p^\prime \\
        &= q + \inn{p,q}p - \inn{p^\prime, q} (p\vee q - p) \\
        &= q + (\inn{p,q}p + \inn{p^\prime, q})p - \inn{p^\prime, q} (p\vee q) \\
        &= q + \inn{p\vee q,q}p + \inn{p\vee q - p, q} (p\vee q) \\
        &= q+p + (1-\inn{p,q})(p\vee q)
        \end{align*}
        which is indeed symmetric in $p$ and $q$.

        \item When $p\mult q = 0$ we have $q\leq p^\perp$ so that $p^\perp \mult q = q$ which indeed gives $p*q = {\frac12(q + p\mult q - p^\perp \mult q)} = \frac12(q-q) = 0$. For the second point we note that because $p\mult q = 0$ we have $p\commu q, q^\perp$ and $q\commu p^\perp$, and hence that the maps $L_p, L_{p^\perp}, L_q$ and $L_{q^\perp}$ commute so that the maps $b\mapsto p*b$ and $b\mapsto q*b$ will commute as well.

        \item Follows immediately from $p\mult p = p$ and $p^\perp \mult p =0$. \qedhere
    \end{enumerate}
\end{proof}
As a result of this lemma we can extend the Jordan product by linearity to the entirety of the space.
\begin{definition}
    Let $a,b\in V$ be arbitrary. Let $a=\sum_i \lambda_i p_i$ and $b=\sum_j \mu_j q_j$ be spectral decompositions with the $p_i$ and $q_j$ atomic. Define their Jordan product as $a*b = \sum_{i,j} \lambda_i \mu_j (p_i*q_j)$. We write $T_a:V\rightarrow V$ for the operator that sends $b$ to $a*b$.
\end{definition}
\begin{proposition}\label{prop:jordanproduct}
    The Jordan product is well-defined, bilinear, commutative and furthermore
    \begin{enumerate}[label=\arabic*., ref=\arabic{counter}.\arabic*]
        \item \label{prop:jordancomm1} If $a\commu b$ then $T_aT_b = T_bT_a$.
        \item \label{prop:jordancomm} If $a\commu b$ then $T_a b = a^+\mult b - a^- \mult b$ where $a^+$ and $a^-$ are the unique orthogonal positive elements such that $a=a^+ - a^-$.
    \end{enumerate}
\end{proposition}
\begin{proof}
    Write $a$ as a spectral decomposition into atomic effects: $a=\sum_i \lambda_i p_i$. We first note that we of course have $a*b = \sum_i \lambda_i p_i*b$ so that the definition is independent of how $b$ is represented as a sum of atomic sharp effects. By the previous lemma the product is commutative and therefore we see it is bilinear and well-defined.
    \begin{enumerate}
        \item Suppose $a\commu b$. By considering the classical algebra spanned by both $a$ and $b$ we can find an orthogonal set of atomic sharp effects $p_i$ such that $a=\sum_i \lambda_i p_i$ and $b=\sum_i \mu_i p_i$. Since $p_i*p_j=0$ for $i\neq j$ we have $p_i*(c*p_j) = p_j*(c*p_i)$ by Lemma~\ref{lem:atomiccommute} so that we can then write
        $$T_bT_a c = b*(a*c) = \sum_{i,j} \mu_i\lambda_j p_i*(c*p_j) = \sum_{i,j} \mu_i \lambda_j p_j*(c*p_i) = a*(c*b) = T_aT_b c$$
        which holds for all $c$. We conclude that $T_bT_a=T_aT_b$.

        \item If atomic $p_i$ commutes with $b$, then $p_i*b = \frac{1}{2}(\id + U_{p_i} - U_{p_i^\perp})b = \frac{1}{2}(b + p_i\mult b - p_i^\perp b) = \frac{1}{2}((p_i+p_i^\perp)\mult b + p_i\mult b - p_i^\perp b) = p_i\mult b$. So by writing $a=\sum_i \lambda_i p_i = \sum_{i, \lambda_i>0} \lambda p_i - \sum_{i, \lambda_i<0}\lvert \lambda_i\rvert p_i = a^+ - a^-$ with $p_i\commu b$ the desired result follows by linearity. \qedhere
    \end{enumerate}
\end{proof}

\begin{theorem}\label{theor:seqprodisjordan}
    Let $V$ be a finite-dimensional sequential product space, then it is a Euclidean Jordan algebra with the Jordan product as defined above.
\end{theorem}
\begin{proof}
    We have already established that the product is bilinear and commutative. Note that since $a\commu a$ we get $a*a = a^+\mult a - a^-\mult a = (a^+)^2 + (a^-)^2 = a^2$ using Proposition~\ref{prop:jordancomm}. Because of course $a\commu a^2$ we get $T_aT_{a*a} = T_{a*a}T_a$ as a consequence of Proposition~\ref{prop:jordancomm1} so that the Jordan identity holds. Since $a*a = a^2\geq 0$ we also see that the algebra is \emph{formally real}: if $\sum_i a_i*a_i = 0$ then for all $i$: $a_i=0$. It is a well-known result (see for instance~\cite[Proposition VIII.4.2]{faraut1994analysis}) that if a Jordan product is formally real, that the algebra is Euclidean, with the product being symmetric with regards to the (essentially unique) self-dual inner product.
\end{proof}

\begin{remark}
    It is also possible to show in a more direct manner that the Jordan product is symmetric with respect to the inner product. We sketch here how to do so. First it must be established that $L_a$ for invertible $a$ commute with their adjoints $L_a^*$ by exploiting the fact that $\Theta = L_{a^{-1}}L_a^*$ must be a unital order-isomorphism necessarily satisfying $\Theta^{-1} = \Theta^*$. Because the mapping $a\mapsto L_a$ is continuous, the result extends to all $a$ and hence it holds in particular for $L_p$ with $p$ sharp. It is a standard result that an idempotent map that commutes with its adjoint is in fact self-adjoint. Since the Jordan product is defined as a linear combination of product maps of sharp effects this indeed establishes the desired result.
\end{remark}
    
\section{Local Tomography and C*-algebras}\label{sec:loctom}
In this section we will let $V$ and $W$ be finite-dimensional sequential product spaces, and hence by the previous sections, they will be Euclidean Jordan algebras. For the duration of this section we will assume that they have a locally tomographic composite (see Definition~\ref{def:localcomposite}), i.e.\ that their linear algebraic tensor product $V\otimes W$ is also a sequential product space and that the sequential product satisfies
$$ (a_1\otimes b_1)\mult (a_2\otimes b_2) = (a_1\mult a_2) \otimes (b_1\mult b_2).$$

\noindent Note that by definition of the tensor product, any element of $V\otimes W$ can be written as $\sum_i \lambda_i a_i\otimes b_i$ where $a_i\in V$, $b_i\in W$ and $\lambda_i \in \R$.

\begin{proposition}\label{prop:atomictensor}
    Let $p\in V$ and $q\in W$ be atomic, then $p\otimes q\in V\otimes W$ is also atomic.
\end{proposition}
\begin{proof}
    Because $(p\otimes q)\mult (p\otimes q) = (p\mult p)\otimes (q\mult q) = p\otimes q$ it is sharp. Let $c = \sum_i \lambda_i a_i\otimes b_i$ be an arbitrary element of $V\otimes W$, then using Lemma~\ref{lem:atomicnorm} $(p\otimes q)\mult c = \sum_i \lambda_i (p\mult a_i)\otimes (q\mult b_i) = \sum_i \lambda_i \norm{p\mult a_i}\norm{q\mult b_i} (p\otimes q) = \mu (p\otimes q)$ for some $\mu\in \R$. Since $c$ was arbitrary we conclude that $p\otimes q$ is atomic as a result of Lemma~\ref{lem:atomicnorm}.
\end{proof}

\begin{definition}
    Let $V$ be a sequential product space with effects $E$. We call $c\in V$ \emph{classical} when it is sharp and compatible with all other effects: $a\commu c$ for all $a\in E$. We will call a classical effect \emph{minimal} when there is no non-zero classical effect strictly below it.
\end{definition}

\begin{proposition}\label{prop:classicaltensor}
    Let $c\in V$ and $d\in W$ be classical, then $c\otimes d$ is classical in $V\otimes W$.
\end{proposition}
\begin{proof}
    $c\commu a$ for all $a \in V$ and $d\commu b$ for all $b\in W$, therefore $c\otimes d\commu a\otimes b$ and the same holds for linear combinations of these elements which span the entirety of $V\otimes W$.
\end{proof}

We call an EJA \emph{simple} if it contains no non-trivial classical effects. We can write any EJA uniquely as $E_1\oplus\ldots \oplus E_k$ where the $E_i$ are simple EJAs, which we will refer to as the \emph{summands} of the EJA. An EJA with $k$ summands has exactly $k$ minimal classical effects, corresponding to the units of each of the summands. Each other classical effect is a sum of these minimal ones.

\begin{lemma}\label{lem:atomicoverlap}
    Let $p$ and $q$ be atomic effects in an EJA. If $q\mult p \neq 0$, then $p$ and $q$ belong to the same simple summand.
\end{lemma}
\begin{proof}
    Suppose $c$ is a sharp classical effect. Then $c\mult p = p\mult c = \lambda p$, for some $\lambda\in[0,1]$. But also $\lambda p = c\mult p=(c^2)\mult p = c\mult (c\mult p) = \lambda (c\mult p) = \lambda^2 p$, so that $\lambda = 1$ or $\lambda = 0$. So either $p\leq c$ or $p$ and $c$ are orthogonal. Suppose $p\leq c$, and that $q\mult p \neq 0$, then $0\neq q\mult p \leq q\mult c =\lambda q$, so that $\lambda\neq 0$ and thus $q\leq c$. If we let $c$ be the identity of the summand that $p$ belongs to we see that the desired property follows.
\end{proof}

\begin{lemma}
    Let $p_1,\ldots,p_r$ be a maximal collection of orthogonal non-zero atomic effects in a simple EJA, then there exists an atomic effect $q$ such that $q\mult p_i\neq 0$ for all $i$.
\end{lemma}
\begin{proof}
    We do this by case distinction using the classification of simple EJAs \cite{jordan1993algebraic}. Either the space is a spin-factor, in which case a maximal collection is always given by a sharp atomic $p$ and its complement $p^\perp$. Any $q\neq p,p^\perp$ cannot be orthogonal to them, because the space is of rank 2, so that indeed $q\mult p\neq 0$ and $q\mult p^\perp \neq 0$.

    If the space is not a spin-factor, then it must be of the form $B(H)^\sa$ for a real, complex, quaternionic or octonion finite-dimensional Hilbert space $H$ (in the case of the octonions we must have $\dim H=3$). For such a space the atomic idempotents correspond to unit vectors of the underlying Hilbert space: $p_i = \lvert v_i\rangle\langle v_i \rvert$ where $v_i\in H$ is some unit vector. We can then take $q=\lvert w\rangle\langle w \rvert$ with $w=\frac{1}{\sqrt{r}}\sum_{i=1}^r v_i$. It should then be clear that $\inn{q,p_i}\neq 0$ and hence $q\mult p_i \neq 0$.
\end{proof}

\begin{proposition}\label{prop:tensordirectsum}
    Let $V = E_1\oplus\ldots \oplus E_m$ and $W = F_1\oplus \ldots \oplus F_n$ with the $E_i$ and $F_j$ being simple EJAs. Let $1\leq k \leq n$ and $1\leq l\leq m$. Let $p_1,\ldots, p_r$ be a maximal collection of orthogonal non-zero atomic effects in $E_k$, and let $q_1,\ldots, q_s$ be such a maximal collection in $F_l$. Then $(p_i\otimes q_j)_{i=1,j=1}^{r,s}$ belong to the same simple summand in $V\otimes W$ and they form a maximal collection of orthogonal non-zero atomic effects in this summand.
\end{proposition}
\begin{proof}
    We let $p$ be atomic such that $p\mult p_i\neq 0$ for all $i$ which exists by the previous lemma and similarly we let $q$ be atomic such that $q\mult q_j \neq 0$ for all $j$. By proposition \ref{prop:atomictensor} $p\otimes q$ and $p_i\otimes q_j$ will be atomic for all $i$ and $j$. By construction we of course have $0\neq (p\mult p_i)\otimes (q\mult q_j)=(p\otimes q)\mult (p_i\otimes q_j)$ and by lemma \ref{lem:atomicoverlap} the $p_i\otimes q_j$ must then belong to the summand of $p\otimes q$ for all $i$ and $j$. 

    Since $\sum_i p_i=1_{E_k}$, this sum is a classical effect. The same holds for $\sum_j q_j=1_{F_l}$. Their tensor product $1_{E_k}\otimes 1_{F_l} = \sum_{i,j} p_i\otimes q_j$ is then also classical by proposition \ref{prop:classicaltensor}. Since the only nonzero classical effect in a simple summand is the identity this expresion must be equal to the identity of this summand. As a result the set $(p_i\otimes q_j)_{i,j}$ is indeed maximal.
\end{proof}
Using this proposition we conclude that for each of the summands $E$ of $V$ and $F$ of $W$ there must exist a summand in $V\otimes W$ which has rank $\rnk~ E~ \rnk~ F$. Because the tensor product map is obviously injective this factor must have dimension at least $\dim E~ \dim F$, and then because of local tomography the dimension must be strictly equal. Now let $V$ be a sequential product space for which the tensor product $V\otimes V$ exists so that the above must in particular be true when $E=F$, i.e. if $E$ is a simple factor of $V$ then there must exist a simple factor with rank $(\rnk E)^2$ and dimension $(\dim E)^2$.

\begin{proposition}\label{prop:simpleEJAsquare}
    Let $E$ be a simple Euclidean Jordan algebra of rank $r> 1$ and dimension $N$. There exists a simple Euclidean Jordan algebra of rank $r^2$ and dimension $N^2$ if and only if $E=B(H)^\sa$ where $H$ is a complex finite-dimensional Hilbert space.
\end{proposition}
\begin{proof}
    If $E=B(H)^\sa$ is a complex matrix algebra the property is obviously true by considering $B(H\otimes H)^\sa$ as the simple EJA of rank $r^2$ and dimension $N^2$. We simply check that every other simple EJA is not a possibility. 
    If $E= B(H)^\sa$ where $H$ is the 3-dimensional octonion Hilbert space, then $r=3$ and $N=27$. The highest dimensional simple EJA of rank $9$ is the quaternionic system which has dimension $9*(2*9-1)=153<27^2 = 729$ so that this is not possible. 
    If $E=B(H)^\sa$ with $H$ quaternionic, then $N = r(2r-1)$. The highest dimensional simple EJA of rank $r^2$ is also quaternionic so that its dimension is $r^2(2r^2-1)$. It is easy to check that $N^2 = r^2(2r-1)^2 > r^2(2r^2-1)$ when $r>1$ so that again, it cannot be this space. 
    If $E=B(H)^\sa$ where $H$ is real, then by dimension counting we can again see that there does not exist an EJA with rank $r^2$ and dimension $N^2$.
    A spin factor always has rank 2. The rank 4 EJAs have dimension 10, 16 and 28. The only one of these which is a square is 16. The 4 dimensional spin-factor corresponds to the qubit which is indeed $B(H)^\sa$ with $H$ a 2-dimensional complex Hilbert space.
\end{proof}

\begin{theorem}\label{theor:seqprodlocalcomp}
    Suppose $V$ is a finite-dimensional sequential product space which as a locally tomographic composite with itself, i.e. for which the linear algebraic tensor product $V\otimes V$ is also a sequential product space satisfying $(a\otimes b)\mult (c\otimes d) = (a\mult c)\otimes (b\mult d)$ for all $a,b,c,d \in V$. Then there exists a C$^*$-algebra $A$ such that $V\cong A^{\text{sa}}$.
\end{theorem}
\begin{proof}
    As established, $V$ is a Euclidean Jordan algebra. As a result of Proposition~\ref{prop:tensordirectsum} for each summand of $V = E_1\oplus\ldots\oplus E_n$ there must exist a simple EJA of rank $(\rnk E_i)^2$ and dimension $(\dim E_i)^2$. By Proposition~\ref{prop:simpleEJAsquare} this is only possible if $E_i=B(H)^\sa$ where $H$ is a complex Hilbert space. Therefore $V$ is a direct sum of complex matrix algebras which means it is the set of self-adjoint elements of a C$^*$-algebra.
\end{proof}

This result should be compared to a result by Hanche-Olsen that states that when $V\otimes B(\mathbb{C}^2)^{\text{sa}}$ is a JB-algebra (see below for the definition) for $V$ a JB-algebra, satisfying $(a\otimes b)*(c\otimes d) = (a*c)\otimes (b*d)$, then $V$ is the self-adjoint part of a C$^*$-algebra~\cite{hanche1985jb}.

\section{Infinite-dimensional sequential product spaces} \label{sec:infiniterank}
\noindent In order to state an infinite-dimensional generalisation of Theorem~\ref{theor:seqprodisjordan} we must first give an appropriate definition of infinite-dimensional Jordan algebras.

\begin{definition}
    An order unit space $V$ is a \emph{JB-algebra} when it is complete in its norm topology and it has a Jordan product $*$ such that $\norm{a*a} = \norm{a}^2$ and $\norm{a*b} \leq \norm{a}\norm{b}$ for all $a,b \in V$.
\end{definition}

As the proof in the finite-dimensional case relied heavily on using atomic effects, which do not have to exist in infinite dimension, we will need to require a few other properties. In particular, we will assume that the order structure on the space is quite rich:

\begin{definition}
    Let $V$ be an order unit space. A subset $S\subseteq V$ is called \emph{bounded} when there exists $r\in \R$ such that $\norm{a}\leq r$ for all $a\in S$. We say $V$ is \emph{$\sigma$-complete} when every bounded increasing sequence $a_1\leq a_2\leq a_3\leq \ldots$ has a supremum.
\end{definition}

\begin{definition}
    A state $\omega: V\rightarrow \R$ on an order unit space $V$ is called \emph{$\sigma$-normal} when it preserves suprema of bounded increasing sequences: $\omega(\bigvee_i a_i) = \bigvee_i \omega(a_i)$. We say that $V$ has \emph{enough normal states} when the $\sigma$-normal states separate the elements, i.e.\ when $\omega(v) = \omega(w)$ for all $\sigma$-normal states $\omega$ implies that $v = w$.
\end{definition}

Finally, we slightly modify the definition of sequential product in order to be more aligned to this order structure:

\begin{definition}
    A $\sigma$-sequential product space is a (possibly infinite-dimensional) $\sigma$-complete order unit space $V$ with enough normal states that comes equipped with a map \\ ${\&:[0,1]_V\times [0,1]_V \rightarrow [0,1]_V}$ satisfying for all $a,b,c \in [0,1]_V$:
    \begin{enumerate}[label=({T}\theenumi), ref=T\theenumi]
        \item Additivity: $a\mult (b+c) = a\mult b+ a\mult c$.
        \item Normality: The map $b\mapsto a\mult b$ is $\sigma$-normal: $a\mult(\bigvee_i b_i) = \bigvee_i (a\mult b_i)$ for an increasing sequence $\{b_i\}$.
        \item Unitality: $1\mult a = a$.
        \item Compatibility of orthogonal effects: If $a\mult b = 0$ then also $b\mult a =0$.
        \item Associativity of compatible effects: If $a\commu b$ then $a\mult (b\mult c) = (a\mult b)\mult c$.
        \item Additivity of compatible effects: If $a\commu b$ then $a \commu 1-b$, and if also $a\commu c$ then $a\commu (b+c)$.
        \item Multiplicativity of compatible effects: If $a\commu b$ and $a\commu c$ then $a\commu (b\mult c)$.
    \end{enumerate}
    We call this map the \emph{sequential product} of the space.
\end{definition}

\begin{remark}
    The definition of the sequential product operation is exactly the one of Definition~\ref{def:seqprod}, except that we require the operation to be normal in the second argument, instead of being norm-continuous in the first argument. The definition here corresponds exactly to the $\sigma$-sequential product of Ref.~\cite{gudder2002sequential}.
\end{remark}

\noindent For the remainder of the section we will let $V$ denote a $\sigma$-sequential product space and $\&$ its sequential product. We conjecture that each $\sigma$-sequential product space is a JB-algebra, but we do not seem to have the right tools to prove this. In order to show a correspondence of $\sigma$-sequential product spaces to $\sigma$-complete JB-algebras, we will need to assume two additional properties on the sequential product.

The first is inspired by the notion of \emph{comprehension} from effectus theory~\cite{cho2015introduction,cho2015quotient}:
\begin{definition}
    We say the sequential product is \emph{comprehensive} when for all sharp effects $q$ the following implication holds for all $\sigma$-normal states $\omega$: if $\omega(q) = 1$, then $\omega(q\mult p) = \omega(p)$ for any $p$.
\end{definition}

What this property says is that if an effect $q$ already holds with certainty on a state $\omega$, then measuring $q$ does not effect the probabilities of other effects holding in the state $\omega$. If the sequential product of $V$ is comprehensive, then every effect of $V$ has a comprehension, as defined in Ref.~\cite{cho2015introduction}, hence the name.

The second property is a weaker version of the \emph{fundamental equality of quadratic Jordan algebras}~\cite{mccrimmon1966general}.
\begin{definition}
    We say the sequential product is \emph{quadratic} when for any two sharp effects $p$ and $q$ we have $q\mult(p\mult q) = (q\mult p)^2$.
\end{definition}
We unfortunately do not know of a reasonable operational interpretation of this property, but we do note that if we are considering C$^*$-algebras, then this property is evidently true, because $q\mult (p\mult q) = q(pqp)q$ while $(q\mult p)^2 = (qpq)^2$. This property also naturally arises in a \emph{$\dagger$-effectus}~\cite{basthesis}.

We will now work towards showing that $\sigma$-sequential product spaces with a comprehensive quadratic sequential product are JB-algebras. Our proof is based on Theorem 9.84 of Ref.~\cite{alfsen2012geometry}.

\begin{definition}
    We call $a\in V$ \emph{simple} when we can write it as $a=\sum_{i=1}^r \lambda_i p_i$ for some $r\in \N$ where $\lambda_i \in \R$ and the $p_i$ are sharp and orthogonal.
\end{definition}

\begin{proposition}[{cf.~\cite{wetering2018characterisation}}]
    Every element $a\in V$ of a $\sigma$-sequential product space $V$ can be written as the supremum of an increasing sequence of simple elements: $a=\bigvee_i a_i$ where the $a_i$ are simple.
\end{proposition}

\begin{lemma}
    If $a\commu b$, then $a^2 - b^2 = (a+b)\mult (a-b)$.
\end{lemma}
\begin{proof}
    We simply calculate $(a+b)\mult (a-b) = (a+b)\mult a - (a+b) \mult b = a\mult (a+b) - b\mult (a+b) = a\mult a + b\mult a - a\mult b - b\mult b = a\mult a - b\mult b$.
\end{proof}
\begin{lemma}
    Let $q$ and $p$ be sharp effects, then $(q\mult p)^2 - (q\mult p^\perp)^2 = q\mult p - q\mult p^\perp$.
\end{lemma}
\begin{proof}
    Since $q\mult p \commu q$ and $q\mult p^\perp = q - q\mult p$ we see that $q\mult p \commu q\mult p^\perp$ and hence by the previous lemma where $a=q\mult p$ and $b=q\mult p^\perp$ we get $(q\mult p)^2 - (q\mult p^\perp)^2 = (q\mult p + q\mult p^\perp)\mult (q\mult p - q\mult p^\perp) = q\mult (q\mult (p-p^\perp)) = q\mult p - q\mult p^\perp$
\end{proof}

\begin{proposition}
    Suppose the sequential product is comprehensive and quadratic, then for any sharp effect $q$ the following implication holds: if $\omega(q)=0$ then $\omega(p\mult q) = \omega(p^\perp \mult q)$ for any sharp effect $p$.
\end{proposition}
\begin{proof}
    Suppose $\omega(q)=0$, then of course $\omega(q^\perp)=1$. Notice that it suffices to prove the equality $q^\perp\mult(p\mult q) = q^\perp\mult(p^\perp\mult q)$ since then, by using the comprehension property, $\omega(p\mult q) = \omega(q^\perp\mult(p\mult q)) = \omega(q^\perp\mult(p^\perp \mult q)) = \omega(p^\perp \mult q)$. By adding the expression $q^\perp\mult(p\mult q^\perp) + q^\perp\mult(p^\perp\mult q^\perp)$ to both sides of the equality we see that it is equivalent to proving the equality
    $$ q^\perp\mult(p\mult (q+q^\perp)) + q^\perp\mult(p^\perp\mult q^\perp) = q^\perp\mult(p^\perp\mult (q+q^\perp)) + q^\perp\mult(p\mult q^\perp).$$
    We then use the quadratic property to write $q^\perp\mult(p^\perp \mult q^\perp) = (q^\perp \mult p^\perp)^2$ and similarly with $p^\perp$ replaced with $p$ to simplify this expression to
    $$ q^\perp \mult p + (q^\perp\mult p^\perp)^2 = q^\perp\mult p^\perp + (q^\perp \mult p)^2.$$
    The desired equality is therefore proven if we can show that
    $$(q^\perp \mult p)^2 - (q^\perp\mult p^\perp)^2 = q^\perp \mult p - q^\perp \mult p^\perp,$$
    but this follows immediately from the previous lemma.
\end{proof}
\begin{proposition}
    Suppose the sequential product is comprehensive and quadratic, then for any $a\geq 0$ the following implication holds for any $\sigma$-normal $\omega$: if $\omega(a)=0$ then $\omega(p\mult a) = \omega(p^\perp \mult a)$ for any sharp effect $p$.
\end{proposition}
\begin{proof}
    First suppose $a$ is simple: $a=\sum_i \lambda_i q_i$. Since $\omega(a)=0$ we of course also have $\omega(q_i)=0$ and hence by the previous proposition $\omega(p\mult q_i) = \omega(p^\perp \mult q_i)$. Then by linearity $\omega(p\mult a) = \sum_i\lambda_i\omega(p\mult q_i) = \sum_i \lambda_i\omega(p^\perp\mult q_i) = \omega(p^\perp \mult a)$. Any $a\geq 0$ can be written as $a=\vee_i a_i$ where the $a_i$ are simple and $a_i\leq a$. Suppose $\omega(a)=0$, then also $\omega(a_i)=0$, and hence: $\omega(p\mult a) = \omega(p\mult \bigvee_i a_i) = \omega(\bigvee_i p\mult a_i) = \bigvee_i \omega(p\mult a_i) = \bigvee_i \omega(p^\perp \mult a_i) = \omega(p^\perp \mult \bigvee_i a_i) = \omega(p^\perp \mult a)$.
\end{proof}

\begin{theorem}
    A $\sigma$-sequential product space with a comprehensive and quadratic sequential product is a $\sigma$-complete JB-algebra.
\end{theorem}
\begin{proof}
    We follow the proof of~\cite[Theorem 9.48]{alfsen2012geometry}.
    Let $a$ be an effect and let $p$ be sharp. We let $L_p$ denote the multiplication operator of $p$. Let $D_p:= L_p - L_{p^\perp}$. By the previous proposition when $\omega(a)=0$ we have $\omega(D_p a) = \omega(p\mult a) - \omega(p^\perp\mult a) = 0$. Because $V$ is $\sigma$-complete, $V$ is also complete in the norm~\cite{jacobs2017distances} and hence by~\cite[Proposition 1.108]{alfsen2012state} the linear map $D_p$ is then an \emph{order derivation}, i.e.\ the exponential map $\exp(t D_p)$ is an order isomorphism for any $t\in \R$. 
    Using exactly the same proof as in~\cite[Theorem 9.48]{alfsen2012geometry} we can then show that for any sharp $p$ and $q$ we will have $[D_p,D_q]1 = 0$ and hence also $[T_p, T_q]1 = 0$ where $T_p := \frac12(\id + D_p)$. As $T_p1 = p$, this commutator being zero tells us that $T_p q = T_q p$. Hence we can define the Jordan product of sharp effects as $p*q := T_p q = T_q p$. 
    As this is symmetric, and linear in one of the arguments, we can easily extend this by bilinearity to all simple elements. It is then straightforward to show that for simple elements $a*a=a^2$ and as a consequence $\norm{a*b} \leq \norm{a}\norm{b}$ and hence this product is continuous in both arguments. We can then extend this product by continuity to the entirety of the space (as the simple elements lie dense in the space). 
    That it satisfies the Jordan equality is shown in the same way as in Theorem~\ref{theor:seqprodisjordan}. We then have shown that the space is a JB-algebra and it is by assumption $\sigma$-complete.
\end{proof}

\section{Minimality of axioms}\label{sec:axioms}

In this section we will discuss the minimality of the conditions and the axioms needed to show that finite-dimensional sequential product spaces are Euclidean Jordan algebras. For easy reference we copy Definition~\ref{def:seqprod} here:

\begin{definition*}
    A map ${\&:[0,1]_V\times [0,1]_V \rightarrow [0,1]_V}$ is called a sequential product when it satisfies the following properties for all $a,b,c \in [0,1]_V$:
    \begin{enumerate}[label=({S}\theenumi), ref=S\theenumi]
        \item Additivity: $a\mult (b+c) = a\mult b+ a\mult c$.
        \item Continuity: The map $a\mapsto a\mult b$ is continuous in the norm.
        \item Unitality: $1\mult a = a$.
        \item Compatibility of orthogonal effects: If $a\mult b = 0$ then also $b\mult a =0$.
        \item Associativity of compatible effects: If $a\commu b$ then $a\mult (b\mult c) = (a\mult b)\mult c$.
        \item Additivity of compatible effects: If $a\commu b$ then $a \commu 1-b$, and if also $a\commu c$ then $a\commu (b+c)$.
        \item Multiplicativity of compatible effects: If $a\commu b$ and $a\commu c$ then $a\commu (b\mult c)$.
    \end{enumerate}
\end{definition*}

\noindent First of all, for the proofs of the main theorems, axiom \ref{ax:compmult} is actually not needed since the following weaker version is sufficient:
\begin{proposition}
    Suppose $a\commu b,c$ and that $b\commu c$, then $a\commu (b\mult c)$.
\end{proposition}
\begin{proof}
    Using axiom \ref{ax:assoc} repeatedly:
    $a\mult (b\mult c) = (a\mult b)\mult c = (b\mult a)\mult c = b\mult (a\mult c) = b\mult (c\mult a) = (b\mult c)\mult a$.
\end{proof}
The reason we included axiom \ref{ax:compmult} is because it is part of the definition of a sequential effect algebra, and because when defining the classical algebra of an element when working in infinite dimension, it is needed to show that the algebra is closed under multiplication.

Of the other axioms, the ones that seem less essential are \ref{ax:orth} and \ref{ax:cont}, so it would be interesting to see what can be done without them.

To define and study the classical algebra of an effect, \ref{ax:orth} is not needed and \ref{ax:cont} is only needed to show that $(\lambda 1)\mult a = \lambda a$. The spectral theorem and the homogeneity of the space can thus be proven without using these axioms if \ref{ax:unit} is changed to $(\lambda 1)\mult a = \lambda a$. When restricting to rank 2 spaces, axioms \ref{ax:add}, \ref{ax:unit}, \ref{ax:assoc} and \ref{ax:compadd} are then sufficient to prove that the space is a Euclidean Jordan algebra (or specifically, a spin-factor). Since the spectral theorem is also what is needed to show that $L_a$ for $a$ invertible is an order-isomorphism, it should be clear that on an EJA this restricted set of axioms already greatly reduces the possible sequential-product-like maps. 

Note also that using the T-algebra formalism of Vinberg \cite{vinberg1967theory} it is possible to find an associative binary operation (see the beginning of section 4 of \cite{chua2003relating} for this operation) for the positive elements in any finite-dimensional homogeneous space that satisfies axioms \ref{ax:add}, \ref{ax:cont}, \ref{ax:unit}, \ref{ax:assoc} (and by associativity also \ref{ax:compmult}) but this product does not satisfy axioms \ref{ax:orth} and \ref{ax:compadd}. This actually leads to an interesting observation: if either the proof of homogeneity in section \ref{sec:prelim} can be shown to hold without use of axiom \ref{ax:compadd} or if a binary product on homogeneous spaces can be found that also satisfies \ref{ax:compadd}, then this would give us a new characterisation of homogeneous spaces. In particular, in the second case, it would show that \ref{ax:orth} is the key to establishing self-duality.

When one considers more general ordered vector spaces than order unit spaces, one can find non-trivial totally ordered vector spaces that allow a commutative bilinear product, and hence a sequential product \cite{basmaster}. These spaces are pathological in the sense that they have `infinitesimal' effects, i.e.\ effects than cannot be distinguished using states.

\section*{Acknowledgements}
\noindent The author would like to thank Bas and Bram Westerbaan for all the useful and insightful conversations regarding effect algebras and order unit spaces and Alex Kolmus and Ema Alsina for suggestions to improve the manuscript. This work is supported by the ERC under the European Union’s Seventh Framework Programme (FP7/2007-2013) / ERC grant n$^\text{o}$ 320571.

\bibliography{bibliography}

\end{document}